\documentclass[12pt,letterpaper,showkeys]{article} 
\usepackage{graphicx}
\usepackage{amssymb}
\usepackage{amsmath}
\usepackage{mathrsfs}
\newtheorem{theorem}{Theorem}[section]

\newtheorem{proposition}[theorem]{Proposition} 
\newtheorem{remark}[theorem]{Remark}
\newtheorem{assumption}{Assumption} 
\newtheorem{definition}{Definition}
\newcommand{\beq}{\begin{equation}} 
\newcommand{\eeq}{\end{equation}} 
\newcommand{\beqa}{\begin{eqnarray}} 
\newcommand{\eeqa}{\end{eqnarray}} 
\newcommand{\beqas}{\begin{eqnarray*}} 
\newcommand{\eeqas}{\end{eqnarray*}} 
\newcommand{\ba}{\begin{array}} 
\newcommand{\ea}{\end{array}} 
\newcommand{\bi}{\begin{itemize}} 
\newcommand{\ei}{\end{itemize}} 
\newcommand{\gap}{\hspace*{2em}} 
 
\newcommand{\nn}{\nonumber}

\def\eqnok#1{(\ref{#1})} 
\def\argmin{{\rm argmin}}

\def\vgap{\vspace*{.1in}}
\def\QED{\ifhmode\unskip\nobreak\fi\ifmmode\ifinner\else\hskip5pt\fi\fi
  \hbox{\hskip5pt\vrule width5pt height5pt depth1.5pt\hskip1pt}} 

\def\amin{{\alpha_{\min}}}
\def\amax{{\alpha_{\max}}}
\def\arho{{\alpha_{\rho}}}

\def\bbeta{{\bar \beta}}
\def\brho{{\beta_{\rho}}}

\def\cO{{\cal O}}
\def\cS{{\cal S}} 

\def\cX{{\cal X}}

\def\cU{{\cal U}} 
\def\cZ{{\cal Z}}

\def\Diag{{\rm Diag}} 
\def\eps{{\epsilon}}
\def\epso{{\epsilon_o}}

\def\hbeta{{\beta}}

\def\lmin{{\lambda_{\min}}}
\def\lmax{{\lambda_{\max}}}

\def\setX{{\cal X}}

\def\tX{{\tilde X}}

\def\tr{{\rm Tr}}

\def\xhbuk{{X_{\hbeta}(U_k)}}

\title{Adaptive First-Order Methods for General \\ Sparse Inverse Covariance 
Selection
} 
\author{Zhaosong Lu%
\thanks{
Department of Mathematics, Simon Fraser University, Burnaby, BC, 
V5A 1S6, Canada. (email: {\tt zhaosong@sfu.ca}).
This author was supported in part by SFU President's Research Grant
and NSERC Discovery Grant.}
}

\date{December 2, 2008 }

\begin{document}

\maketitle

\begin{abstract}
 
In this paper, we consider estimating sparse inverse covariance of a Gaussian 
graphical model whose conditional independence is assumed to be {\it partially} 
known. Similarly as in \cite{DaOnEl06}, we formulate it as an $l_1$-norm 
penalized maximum likelihood estimation problem. Further, we propose an algorithm 
framework, and develop two first-order methods, that is, the adaptive spectral 
projected gradient (ASPG) method and the adaptive Nesterov's smooth (ANS) method, for 
solving this estimation problem. Finally, we compare the performance of these two 
methods on a set of randomly generated instances. Our computational results demonstrate 
that both methods are able to solve problems of size at least a thousand and 
number of constraints of nearly a half million within a reasonable amount of 
time, and the ASPG method generally outperforms the ANS method.

\vskip14pt

 
\vskip14pt


\end{abstract}

\section{Introduction} \label{introduction}

It is well-known that sparse undirected graphical models are capable of describing 
and explaining the relationships among a set of variables. Given a set of random 
variables with Gaussian distribution, the estimation of such models involves finding 
the pattern of zeros in the inverse covariance matrix since these zeros correspond 
to conditional independencies among the variables. In recent years, a variety of 
approaches have been proposed for estimating sparse inverse covariance matrix. 
(All notations used below are defined in Subsection \ref{notation}.) 
Given a sample covariance matrix $\Sigma \in \cS^n_+$, d'Aspremont 
et al.\ \cite{DaOnEl06} formulated sparse inverse covariance selection as the 
following $l_1$-norm penalized maximum likelihood estimation problem:
\beq \label{relax1}
\max\limits_X\  \{\log\det X - \langle \Sigma, X \rangle - \rho e^T |X| e: \ 
X \succeq 0\},
\eeq
where $\rho>0$ is a parameter controlling the trade-off between likelihood and 
sparsity of the solution. They also studied Nesterov's smooth approximation 
scheme \cite{Nest05-1} and block-coordinate descent (BCD) method for solving 
\eqnok{relax1}. Independently, Yuan and Lin \cite{YuLi07-1} proposed a similar 
estimation problem to \eqnok{relax1} as follows:
\beq \label{relax2}
\max\limits_X \ \{\log\det X - \langle \Sigma, X \rangle - \rho 
\sum\limits_{i\neq j} |X_{ij}|: \  X \succeq 0\}.
\eeq
They showed that problem \eqnok{relax2} can be suitably solved by the interior 
point algorithm developed in Vandenberghe et al.\ \cite{VaBoWu98}. As demonstrated 
in \cite{DaOnEl06,YuLi07-1}, the estimation problems \eqnok{relax1} and \eqnok{relax2} 
are capable of discovering effectively the sparse structure, or equivalently, the 
conditional independence in the underlying graphical model. Recently, 
Lu \cite{Lu07} proposed a variant of Nesterov's smooth method \cite{Nest05-1} 
for problems \eqnok{relax1} and \eqnok{relax2} that substantially outperforms 
the existing methods in the literature. In addition, Dahl et al.\ \cite{DaVaRo06} 
studied the maximum likelihood estimation of a Gaussian graphical model whose 
conditional independence is known, which can be formulated as   
\beq \label{max-likehood}
\max\limits_X \ \{\log\det X - \langle \Sigma, X \rangle: \ X \succeq 0, \ X_{ij}=0, 
\forall (i,j) \in \bar E \},
\eeq  
where $\bar E$ is a collection of all pairs of conditional independent nodes. 
They showed that when the underlying graph is nearly-chordal, Newton's method 
and preconditioned conjugate gradient method can be efficiently applied to solve 
\eqnok{max-likehood}.

In practice, the sparsity structure of a Gaussian graphical model is often partially
known from some knowledge of its random variables. In this paper we consider estimating 
sparse inverse covariance of a Gaussian graphical model whose conditional independence 
is assumed to be {\it partially} known in advance (but it can be completely unknown). 
Given a sample covariance matrix $\Sigma \in \cS^n_+$, we can naturally formulate it 
as the following constrained $l_1$-norm penalized maximum likelihood estimation problem:
\beq \label{gsparcov}
\ba{rl}
\max\limits_X & \log\det X - \langle \Sigma, X \rangle - \sum\limits_{(i,j)\notin\Omega} 
\rho_{ij} |X_{ij}|, \\ [10pt]
\mbox{s.t.} &  X \succeq 0, \ X_{ij}=0, \forall (i,j) \in \Omega,
\ea 
\eeq   
where $\Omega$ consists of a set of pairs of conditionally independent nodes, and $\{\rho_{ij}\}_{(i,j) 
\notin \Omega}$ is a set of nonnegative parameters controlling the trade-off 
between likelihood and sparsity of the solution. It is worth mentioning that unlike in 
\cite{DaVaRo06}, we do not assume any specific structure on the sparsity of underlying 
graph for problem \eqnok{gsparcov}. We can clearly observe that (i) $(i,i) \notin \Omega$ 
for $1 \le i \le n$, and $(i,j) \in \Omega$ if and only if $(j,i) \in \Omega$; (ii) 
$\rho_{ij}= \rho_{ji}$ for any $(i,j)\notin \Omega$; and (iii) problems 
\eqnok{relax1}-\eqnok{max-likehood} can be viewed as special cases of problem 
\eqnok{gsparcov} by choosing appropriate $\Omega$ and $\{\rho_{ij}\}_{(i,j) \notin \Omega}$. 
For example, if setting $\Omega=\emptyset$ and $\rho_{ij}=\rho$ for all $(i,j)$, problem 
\eqnok{gsparcov} becomes \eqnok{relax1}.

It is easy to observe that problem \eqnok{gsparcov} can be reformulated as a constrained 
smooth convex problem that has an explicit $\cO(n^2)$-logarithmically homogeneous 
self-concordant barrier function. Thus, it can be suitably solved by interior point (IP) 
methods (see Nesterov and Nemirovski \cite{NeNe94} and Vandenberghe et al.\ \cite{VaBoWu98}). 
The worst-case iteration complexity of IP methods for finding an $\epsilon$-optimal solution 
to \eqnok{gsparcov} is $\cO(n\log(\epsilon_0/\epsilon))$, where $\epsilon_0$ is an initial 
gap. Each iterate of IP methods requires $\cO(n^6)$ arithmetic cost for assembling and 
solving a typically dense Newton system with $\cO(n^2)$ variables. Thus, the total worst-case 
arithmetic cost of IP methods for finding an $\epsilon$-optimal solution to \eqnok{gsparcov} 
is $\cO(n^7\log(\epsilon_0/\epsilon))$, which is prohibitive when $n$ is relatively large.
  
Recently, Friedman et al.\ \cite{FriHasTib07} proposed a gradient type method for solving 
problem \eqnok{gsparcov}. They first converted \eqnok{gsparcov} into the following 
penalization problem 
\beq \label{penalty-prob1}
\max\limits_{X \succeq 0} \ \log\det X - \langle \Sigma, X \rangle - 
\sum\limits_{i,j} \rho_{ij} |X_{ij}|.   
\eeq
by setting $\rho_{ij}$ to an extraordinary large number (say, $10^9$) for all $(i,j)\in\Omega$. 
Then they applied a slight variant of the BCD method \cite{DaOnEl06} to the dual problem of 
\eqnok{penalty-prob1} in which each iteration is solved by a coordinate descent approach to 
a lasso ($l_1$-regularized) least-squares problem. Given that their method is a gradient type 
method and the dual problem of \eqnok{penalty-prob1} is highly ill-conditioned for the 
above choice of $\rho$, it is not surprising that their method converges extremely slowly. 
Moreover, since the associated lasso least-squares problems can only be solved inexactly, 
their method often fails to converge even for a small problem. 

In this paper, we propose adaptive first-order methods for problem \eqnok{gsparcov}. 
Instead of solving \eqnok{penalty-prob1} once with a set of huge penalty parameters 
$\{\rho_{ij}\}_{(i,j)\in \Omega}$, our methods consist of solving a sequence of problems 
\eqnok{penalty-prob1} with a set of moderate penalty parameters $\{\rho_{ij}\}_{(i,j)\in 
\Omega}$ that are adaptively adjusted until a desired approximate solution is found. 
For a given $\rho$,  problem \eqnok{penalty-prob1} is solved by the adaptive spectral 
projected gradient (ASPG) method and the adaptive Nesterov's smooth (ANS) method that 
are proposed in this paper. 
                   
The rest of paper is organized as follows. In Subsection \ref{notation}, we introduce 
the notations used in this paper. In Section \ref{methods}, we propose an algorithm 
framework and develop two first-order methods, that is, the ASPG and ANS 
methods, for solving problem \eqnok{gsparcov}. The performance of these two methods 
are compared on a set of randomly generated instances in Section \ref{comp}. Finally, 
we present some concluding remarks in Section \ref{concl-remark}. 

\subsection{Notation} \label{notation}

In this paper, all vector spaces are assumed to be finite dimensional. 
The symbols $\Re^n$, $\Re^n_+$ and $\Re^n_{++}$ denote the $n$-dimensional Euclidean 
space, the nonnegative orthant of $\Re^n$ and the positive orthant of $\Re^n$, 
respectively. The set of all $m \times n$ matrices with real entries is denoted by 
$\Re^{m \times n}$. The space of symmetric $n \times n$ matrices will be denoted 
by $\cS^n$. If $X \in \cS^n$ is positive semidefinite, we write $X \succeq 0$. Also, 
we write $X \preceq Y$ to mean $Y-X \succeq 0$. The cone of positive semidefinite 
(resp., definite) matrices is denoted by $\cS^n_+$ (resp., $\cS^n_{++}$). Given
matrices $X$ and $Y$ in $\Re^{m \times n}$, the standard inner product is 
defined by $\langle X, Y \rangle := \tr (XY^T)$, where $\tr(\cdot)$ denotes 
the trace of a matrix. $\|\cdot\|$ denotes the Euclidean norm and its associated 
operator norm unless it is explicitly stated otherwise. The Frobenius norm of a 
real matrix $X$ is defined as $\|X\|_F := \sqrt{\tr(XX^T)}$. We denote by $e$ 
the vector of all ones, and by $I$ the identity matrix. Their dimensions should 
be clear from the context. For a real matrix $X$, we denote by $|X|$ the absolute 
value of $X$, that is, $|X|_{ij}=|X_{ij}|$ for all $i, j$. The determinant and the 
minimal (resp., maximal) eigenvalue of a real symmetric matrix $X$ are denoted by 
$\det X$ and $\lmin(X)$ (resp., $\lmax(X)$), respectively, and $\lambda_i(X)$ 
denotes its $i$th largest eigenvalue. Given an $n \times n$ (partial) matrix 
$\rho$, $\Diag(\rho)$ denotes the diagonal matrix whose $i$th diagonal element 
is $\rho_{ii}$ for $i=1,\ldots,n$. Given matrices $X$ and $Y$ in $\Re^{m \times 
n}$, $X \ast Y$ denotes the pointwise product of $X$ and $Y$, namely, $X \ast Y 
\in \Re^{m \times n}$ whose $ij$th entry is $X_{ij}Y_{ij}$ for all $i, j$. We 
denote by $\cZ_{+}$ the set of all nonnegative integers.      

\section{Adaptive first-order methods}
\label{methods} 

In this section, we discuss some suitable first-order methods for general sparse 
inverse covariance selection problem \eqnok{gsparcov}. In particular, we first 
provide an algorithm framework for it in Subsection \ref{framework}. Then we 
specialize this framework by considering two first-order methods, namely, the 
adaptive spectral projected gradient method and the adaptive Nesterov's smooth 
method in Subsection \ref{special}.

\subsection{Algorithm framework}
\label{framework}
 
In this subsection, we provide an algorithm framework for general sparse inverse 
covariance selection problem \eqnok{gsparcov}.

Throughout this paper, we assume that $\rho_{ij} \ge 0$ is given and fixed for all  
$(i,j)\notin \Omega$, and that the following condition holds.

\begin{assumption} \label{assump}
$\Sigma + \Diag(\rho) \succ 0$.
\end{assumption}

Note that $\Sigma$ is a sample covariance matrix, and hence $\Sigma \succeq 0$. In 
addition, $\Diag(\rho) \succeq 0$. Thus, $\Sigma + \Diag(\rho) \succeq 0$. It may not 
be, however, positive definite in general. But we can always perturb $\rho_{ii}$ by 
adding a small positive number (say, $10^{-8}$) whenever needed to ensure the above 
assumption holds.

We first establish the existence of an optimal solution for problem \eqnok{gsparcov} 
as follows.

\begin{proposition} \label{gsparcov-soln}
Problem \eqnok{gsparcov} has a unique optimal solution $X^* \in \cS^n_{++}$. 
\end{proposition}

\begin{proof}
Since $(i,i) \notin \Omega$ for $i=1,\ldots,n$, we see that $X=I$ is a feasible solution 
of problem \eqnok{gsparcov}. For convenience, let $f(X)$ denote the objective function 
of \eqnok{gsparcov}. We now show that the sup-level set $S_f(I)=\{X \succeq 0:\ f(X) \ge 
f(I), \ X_{ij}=0, \ \forall (i,j)\in\Omega\}$ is compact. Indeed, using the definition 
of $f(\cdot)$, we observe that for any $X\in S_f(I)$, 
\beqas
f(I) & \le & f(X) \ \le \ \log\det X - \langle \Sigma + \Diag(\rho), X \rangle \
\le \ \sum\limits_{i=1}^n \left[\log\lambda_i(X) - \lmin(\Sigma + \Diag(\rho)) 
\lambda_i(X)\right], \\ [6pt]
& \le & (n-1) \left[-1-\log\lmin(\Sigma + \Diag(\rho))\right] + \log\lmax(X) 
- \lmin(\Sigma + \Diag(\rho))\lmax(X),     
\eeqas 
where the last inequality follows from the fact that for any $a > 0$,  
\beq \label{log-max}
\max\limits_t \ \{\log t - a t: \ t \ge 0 \} \ = \ -1 - \log a. 
\eeq 
Hence, we obtain that for any $X\in S_f(I)$, 
\beq \label{log-bdd}
\log\lmax(X) 
- \lmin(\Sigma + \Diag(\rho))\lmax(X) \ \ge \ f(I) - (n-1) \left[-1-\log\lmin(
\Sigma + \Diag(\rho))\right],
\eeq
which implies that there exists some $\beta(\rho) > 0$ such that 
$\lmax(X) \le \beta(\rho)$ for all $X\in S_f(I)$. Thus, $S_f(I) \subseteq 
\{X \in \cS^n: \ 0 \preceq X \preceq \beta(\rho) I\}$. Further, using this result 
along with the definition of $f(\cdot)$, we easily observe that for any 
$X\in S_f(I)$, 
\[
\ba{lcl}
\log\lmin(X) & = & f(X) - \sum\limits_{i=1}^{n-1}\log\lambda_i(X) + \langle \Sigma, 
X \rangle + \sum\limits_{(i,j)\notin\Omega} \rho_{ij} |X_{ij}|, \\ [8pt]
& \ge & f(I) - (n-1) \log\beta(\rho) + \min\limits_{0 \preceq X \preceq \beta(\rho)} 
\{\langle \Sigma, X \rangle + \sum\limits_{(i,j)\notin\Omega} \rho_{ij} |X_{ij}|\}.
\ea
\]
It follows that there exists some $\alpha(\rho)>0$ such that $\lmin(X) \ge \alpha(\rho)$ 
for all $X\in S_f(I)$. Hence, $S_f(I) \subseteq \{X \in \cS^n: \ \alpha(\rho) I \preceq 
X \preceq \beta(\rho) I\}$ is bounded, which together with the fact that $f(\cdot)$ is 
continuous in the latter set, implies that $S_f(I)$ is closed. Therefore, problem 
\eqnok{gsparcov} has at least an optimal solution. Further, observing that $f(\cdot)$ 
is strict concave, we conclude that problem \eqnok{gsparcov} has a unique optimal solution.        
\end{proof}

\vgap

Similarly, we can show that the following result holds.
 
\begin{proposition} \label{penalty-soln}
Given any $\rho_{ij} \ge 0$ for $(i,j) \in \Omega$, problem \eqnok{penalty-prob1} has a 
unique optimal solution $X^* \in \cS^n_{++}$. 
\end{proposition}

Before presenting an algorithm framework for problem \eqnok{gsparcov}, we introduce a 
terminology for \eqnok{gsparcov} as follows.

\begin{definition}
Let $\epso \ge 0$ and $\eps_c \ge 0$ be given. Let $f(\cdot)$ and $f^*$ denote the 
objective function and the optimal value of \eqnok{gsparcov}, respectively. 
$X \in \cS^n_+$  is an $(\epso,\eps_c)$-optimal solution of problem \eqnok{gsparcov} if 
$f(X) \ge f^* - \epso$ and $\max\limits_{(i,j)\in\Omega} |X_{ij}| \le \eps_c$. 
\end{definition}

Analogously, we can define an $\epso$-optimal solution for problem \eqnok{penalty-prob1}. 
Given that our ultimate aim is to estimate a sparse inverse covariance matrix $X^* \succeq 0$ 
that satisfies at least $X^*_{ij}=0$, $\forall (i,j)\in\Omega$ and approximately maximizes 
the log-likelihood, we now briefly discuss how to obtain such an approximate solution $X^*$ 
from an $(\epso,\eps_c)$-optimal solution $\bar X^*$ of \eqnok{gsparcov}. Let us define 
$\tX^*\in\cS^n$ by letting $\tX^*_{ij}=\bar X^*_{ij}$, $\forall (i,j)\notin \Omega$  
and $\tX^*_{ij}=0$, $\forall (i,j)\in\Omega$. We then set $X^* := \tX^*+ t^* I$, where 
\[
t^* = \arg\max\ \{\log\det(\tX^*+t I) - \langle \Sigma, \tX^*+tI\rangle: \ t \ge -
\lambda_{\min}(\tX^*)\}.
\]
It is not hard to see that $t^*$ can be easily found. We also observe that such 
$X^*$ belongs to $\cS^n_{++}$, satisfies $X^*_{ij}=0$, $\forall (i,j)\in\Omega$ and 
retains the same sparsity as $\tX^*$. In addition, by setting the log-likelihood value at 
$\tX^*$ to $-\infty$ if $\lambda_{\min}(\tX^*) \le 0$, we can easily see that the 
log-likelihood value at $X^*$ is at least as good as that at $\tX^*$. Thus, $X^*$ 
is a desirable estimation of sparse inverse covariance, provided $\bar X^*$ is a good 
approximate solution to problem \eqnok{gsparcov}. 

In the remainder of this paper, we concentrate on finding an $(\epso,\eps_c)$-optimal solution 
of problem \eqnok{gsparcov} for any pair of positive $(\epso, \eps_c)$. We next present an 
algorithm framework for \eqnok{gsparcov} based on an adaptive $l_1$ penalty approach.

\gap

\noindent
\begin{minipage}[h]{6.6 in}
{\bf Algorithm framework for general sparse inverse covariance selection  (GSICS):} \\ [5pt]
Let $\epso > 0$, $\eps_c > 0$ and $r_{\rho} > 1$ be given. Let $\rho_{ij}^0 >0, \forall (i,j) 
\in \Omega$ be given such that $\rho_{ij}^0=\rho_{ji}^0, \forall (i,j) \in \Omega$. 
Set $\rho_{ij} = \rho_{ij}^0$ for all $(i,j) \in \Omega$. 
\begin{itemize}
\item[1)] Find an $\epso$-optimal solution $X^{\epso}$ of problem \eqnok{penalty-prob1}.
\item[2)] If $\max\limits_{(i,j)\in\Omega} |X^{\epso}_{ij}| \le \eps_c$, 
terminate. Otherwise, set $\rho_{ij} \leftarrow \rho_{ij} r_{\rho}$ for all $(i,j) \in 
\Omega$, and go to step 1).    
\end{itemize}
\noindent
{\bf end}
\end{minipage}

\gap
     
\begin{remark} 
To make the above framework complete, we need to choose suitable methods for solving problem 
\eqnok{penalty-prob1} in step 1). We will propose first-order methods for it in Subsection 
\ref{special}. In step 2) of the framework GSICS, there are some other strategies for 
updating the penalty parameters $\{\rho_{ij}\}_{(i,j) \in \Omega}$. For example, for 
any $(i,j) \in \Omega$, one can update $\rho_{ij}$ only if $\rho_{ij} > \eps_c$. But 
we observed in our experimentation that this strategy performs worse than the one described 
above. In addition, instead of using a common ratio $r_{\rho}$ for all $(i,j) \in \Omega$, 
one can associate with each $\rho_{ij}$ an individual ratio $r_{ij}$. Also, the ratio 
$r_{\rho}$ is no need to be fixed for all iterations, and it can vary from iteration to 
iteration depending on the amount of violation incurred in $\max\limits_{(i,j)\in\Omega} 
|X^{\epso}_{ij}| \le \eps_c$.      
\end{remark}

\vgap

Before discussing the convergence of the framework GSICS, we first study the 
convergence of the $l_1$ penalty method for a general nonlinear programming (NLP) problem.  

Given a set $\emptyset \neq \setX \subseteq \Re^n$ and functions $f: \setX \to \Re$, 
$g: \setX \to \Re^k$ and $h: \setX \to \Re^l$, consider the NLP problem:
\beq \label{nlp} 
\ba{lcl}
f^* &=& \sup\limits_{x\in\setX}  \ \ f(x) \\
& & \mbox{s.t.} \ \ g(x) = 0, \, \, h(x) \le 0.
\ea 
\eeq
We associate with the NLP problem \eqnok{nlp} the following $l_1$ penalty function:
\beq \label{penalty-fun} 
P(x; \lambda, \mu) \ := \ f(x) - \lambda^T |g(x)| - \mu^T h^+(x),
\eeq
where $\lambda \in \Re^k_+$, $\mu \in \Re^l_+$ and $(h^+(x))_i = \max\{0,h_i(x)\}$ for 
$i=1,\ldots,l$.

We now establish a convergence result for the $l_1$ penalty method for the NLP 
problem \eqnok{nlp} under some assumption on $f(x)$.

\begin{proposition} \label{nlp-converg}
Let $\epso > 0$ and $\eps_c >0$ be given. Assume that there exists some $\bar f \in \Re$ 
such that $f(x) \le \bar f$ for all $x \in \setX$. Let $x^{\epso}_{\lambda,\mu} \in \setX$ 
be an $\epso$-optimal solution of the problem 
\beq \label{penalty-opt}
\sup \ \{P(x; \lambda, \mu): \ x \in \setX\} 
\eeq
for $\lambda\in\Re^k_+$ and $\mu\in\Re^l_+$, and let $v_{\lambda,\mu} := 
\min\{\min\limits_i\lambda_i, \min\limits_i \mu_i\}$. Then $f(x^{\epso}_{\lambda,\mu}) 
\ge f^* - \epso$, and, moreover, $\left\|\left(g(x^{\epso}_{\lambda,\mu}); 
h^+(x^{\epso}_{\lambda,\mu}) \right) \right\|_{\infty} \le \eps_c$ holds whenever 
$v_{\lambda,\mu} \ge (\bar f - f^* + \epso)/\eps_c$, where $f^*$ is the optimal 
value of the NLP problem \eqnok{nlp}.
\end{proposition}

\begin{proof}
In view of the assumption that $f(x)$ is bounded above in $\setX$, we clearly see that 
$f^*$ is finite. Let $f^*_{\lambda,\mu}$ denote the optimal value of problem \eqnok{penalty-opt}. 
We easily observe that $f^*_{\lambda,\mu} \ge f^*$. Using this relation, \eqnok{penalty-fun} 
and the fact that $x^{\epso}_{\lambda,\mu}$ is an $\epso$-optimal solution of \eqnok{penalty-opt}, 
we have 
\beq \label{low-bdd}
f(x^{\epso}_{\lambda,\mu}) \ \ge \ P(x^{\epso}_{\lambda,\mu}; \lambda, \mu) \ \ge \
f^*_{\lambda,\mu} - \epso \ \ge \ f^* - \epso,
\eeq
and hence the first statement holds. We now prove the second statement. Using \eqnok{penalty-fun}, 
\eqnok{low-bdd} and the definition of $v_{\lambda,\mu}$, we have 
\beqa
f(x^{\epso}_{\lambda,\mu}) - v_{\lambda, \mu} \left\|\left(g(x^{\epso}_{\lambda,\mu}); 
h^+(x^{\epso}_{\lambda,\mu})\right)\right\|_{\infty} 
& \ge &
f(x^{\epso}_{\lambda,\mu}) - v_{\lambda, \mu} \left\|\left(g(x^{\epso}_{\lambda,\mu}); 
h^+(x^{\epso}_{\lambda,\mu})\right)\right\|_1 \nn \\
& \ge & P(x^{\epso}_{\lambda,\mu}; \lambda, \mu) \ \ge \ f^* - \epso. \label{penalty-ineq}
\eeqa 
Further, from the assumption, we know $f(x^{\epso}_{\lambda,\mu}) \le \bar f$ 
due to $x^{\epso}_{\lambda,\mu} \in \setX$. This together with \eqnok{penalty-ineq} 
immediately implies that the second statement holds.
\end{proof}

\vgap

We are now ready to establish a convergence result for the framework GSICS.

\begin{theorem} \label{sparcov-converg}
Let $\epso > 0$ and $\eps_c >0$ be given. Suppose that in step 1) of the framework 
GSICS, an $\epso$-optimal solution $X^{\epso}$ of problem \eqnok{penalty-prob1} is 
obtained by some method. Then, the framework GSICS generates an $(\epso,\eps_c)$-optimal 
solution to problem \eqnok{gsparcov} in a finite number of outer iterations, or 
equivalently, a finite number of updates on the penalty parameters $\{\rho_{ij}\}_{(i,j) 
\in \Omega}$.
\end{theorem}

\begin{proof}
Invoking that $\Sigma+\Diag(\rho) \succ 0$ (see Assumption \ref{assump}), we see that 
for any $X\in\cS^n_+$,
\[
\ba{lcl}
\log\det X - \langle \Sigma, X \rangle - \sum\limits_{(i,j) \notin \Omega} \rho_{ij} 
|X_{ij}| & \le & \log\det X - \langle \Sigma+\Diag(\rho), X \rangle  \\ 
& \le & \sup \{\log\det Y - \langle \Sigma+\Diag(\rho), Y \rangle: \ Y \succeq 0\} 
\ < \ \infty,
\ea
\] 
where the last inequality follows from the fact that the above maximization 
problem achieves its optimal value at $Y = (\Sigma+\Diag(\rho))^{-1} \succ 0$. 
This observation together with Proposition \ref{nlp-converg} immediately yields 
the conclusion. 
\end{proof}

\subsection{Adaptive first-order methods for problem \eqnok{penalty-prob1}}
\label{special}

In this subsection, we will discuss some suitable first-order methods for solving 
problem \eqnok{penalty-prob1} that appears in step 1) of the algorithm framework 
GSICS.

As seen from Proposition \ref{penalty-soln}, problem \eqnok{penalty-prob1} has 
a unique optimal solution. We next provide some bounds on it.  

\begin{proposition} \label{soln-bdd}
Let $f_{\rho}(\cdot)$ and $X^*_{\rho}$ denote the objective function and 
the unique optimal solution of problem \eqnok{penalty-prob1}, respectively. 
Let $\vartheta$ be defined as 
\beq \label{vtheta}
\vartheta := \max\left\{f_{\rho}((\Sigma+\Diag(\rho))^{-1}), \theta\right\} 
- (n-1) [-1-\log\lmin(\Sigma+\Diag(\rho))],
\eeq
where $\theta := n (-1-\log\tr(\Sigma+\rho) + \log n)$. Then 
$\arho I \preceq X^*_{\rho}  \preceq \brho I$, where 
$\arho := 1/(\|\Sigma\|+\|\rho\|)$ and $\brho$ is the largest positive 
root of the following equation 
\[
\log t - \lmin(\Sigma+\Diag(\rho)) t - \vartheta = 0. 
\]    
\end{proposition}

\begin{proof}
Let
\beq \label{cU}
\cU := \{U\in \cS^n: \ |U_{ij}| \le 1, \forall ij\},  
\eeq
and 
\beq \label{phi}
\phi(X, U) := \log\det X - \langle \Sigma+\rho \ast U, X \rangle, \ \ \ 
\forall (X,U) \in \cS^n_{++} \times \cU. 
\eeq  
Since $X^*_{\rho} \in \cS^n_{++}$ is the optimal solution of problem 
\eqnok{penalty-prob1}, it can be easily shown that there exists some 
$U^*\in \cU$ such that $(X^*_{\rho}, U^*)$ is a saddle point of 
$\phi(\cdot, \cdot)$ in $\cS^n_{++} \times \cU$, and hence   
\[
X^*_{\rho} = \arg\min\limits_{X \in \cS^n_{++}} \phi(X,U^*).
\]
This relation along with \eqnok{phi} immediately yields 
$X^*_{\rho}(\Sigma + \rho \ast U^*) = I$. Hence, we have   
\[
X^*_{\rho} = (\Sigma + \rho \ast U^*)^{-1} \succeq \frac{1}{\|\Sigma\|+ \|\rho 
\ast U^*\|} I,
\]
which together with \eqnok{cU} and the fact that $U^*\in \cU$, implies 
that $X^* \succeq \frac{1}{\|\Sigma\|+ \|\rho\|} I$. Thus, 
$X^*_{\rho} \succeq \arho I$ as desired. 

We next bound $X^*_{\rho}$ from above. Let $f^*_{\rho}$ denote the optimal 
value of problem \eqnok{penalty-prob1}. In view of the definition of 
$f_{\rho}(\cdot) $ and \eqnok{log-max}, we have 
\[
f^*_{\rho} \ \ge \ \max\limits_{t>0} f_{\rho}(tI) \ = \ \max\limits_{t>0} \ n 
\log t - t \tr(\Sigma + \rho) \ = \ n (-1-\log\tr(\Sigma+\rho) + \log n)  
\ =: \ \theta.
\] 
Thus, $f^*_{\rho} \ge \max\{f_{\rho}((\Sigma+\Diag(\rho))^{-1}), \theta\}$. 
Using this result and following a similar procedure as for deriving 
\eqnok{log-bdd}, we can show that 
\[
\log\lmax(X^*_{\rho}) - \lmin(\Sigma + \Diag(\rho))\lmax(X^*_{\rho}) \ \ge \ 
\vartheta,  
\]
where $\vartheta$ is given in \eqnok{vtheta}, and hence the statement 
$X^*_{\rho} \preceq \brho I$ immediately follows.    
\end{proof}

\vgap

In view of Proposition \ref{soln-bdd}, we see that problem \eqnok{penalty-prob1} 
is equivalent to the following problem
\beq \label{penalty-prob2}
\max\limits_{\arho \preceq X \preceq \brho} \ \log\det X - \langle \Sigma, X 
\rangle - \sum\limits_{i,j} \rho_{ij} |X_{ij}|,   
\eeq
where $\arho$ and $\brho$ are defined in Proposition \ref{soln-bdd}.  

We further observe that problem \eqnok{penalty-prob2} can be rewritten as  
\beq \label{primal}
\max\limits_{X \in \cX_{\rho}} \ \{f_{\rho}(X) := \min\limits_{U \in \cU} \ \phi(X,U)\}, 
\eeq 
where $\cU$ and $\phi(\cdot, \cdot)$ are given in \eqnok{cU} and \eqnok{phi}, 
respectively, and $\cX_{\rho}$ is defined as follows:
\beq \label{cX}
\cX_{\rho} := \{X\in\cS^n: \ \arho I \preceq X \preceq \brho I\}.
\eeq

Observing that $\phi(X,U): \cX_{\rho} \times \cU \to \Re $ is a smooth function 
which is {\it strictly} concave in $X\in \cX_{\rho}$ for every fixed $U\in \cU$, 
and convex in $U\in \cU$ for every fixed $X\in \cX_{\rho}$, we can conclude that   
(i) problem \eqnok{primal} and its dual, that is, 
\beq \label{dual}
\min\limits_{U \in \cU} \ \{g_{\rho}(U):= \max_{X\in\cX_{\rho}} \phi(X,U) \}
\eeq
are both solvable and have the same optimal value; 
and (ii) the function $g_{\rho}(\cdot)$ is convex differentiable and its gradient 
is given by 
\[
\nabla g_{\rho}(U) = \nabla_U \phi(X(U),U), \ \forall U \in \cU,
\]
where 
\beq \label{XU}
X(U) := \arg\max\limits_{X\in\cX_{\rho}} \phi(X,U).
\eeq

The following result shows that the approximate solution of problem \eqnok{primal} 
(or equivalently, \eqnok{penalty-prob1}) can be obtained by solving smooth 
convex problem \eqnok{dual}.

\begin{proposition} \label{pd-soln}
Let $X^*_{\rho}$ be the unique optimal solution of problem \eqnok{primal}, and let 
$f^*_{\rho}$ be the optimal value of problems \eqnok{primal} and \eqnok{dual}. 
Suppose that the sequence $\{U_k\}^{\infty}_{k=0} \subseteq \cU$ is such that 
$g_{\rho}(U_k) \to f^*_{\rho}$ as $k \to \infty$. Then, $X(U_k) \to X^*_{\rho}$ 
and $g_{\rho}(U_k) - f_{\rho}(X(U_k)) \to 0$ as $k \to \infty$, where 
$X(\cdot)$ is defined in \eqnok{XU}.     
\end{proposition} 

\begin{proof}
The proof is similar to that of Theorem 2.4 of Lu \cite{Lu07}.
\end{proof}

\vgap

 From Proposition \ref{pd-soln}, we see that problem \eqnok{penalty-prob1} can 
be solved simultaneously while solving problem \eqnok{dual}. Indeed, suppose 
that $\{U_k\}^{\infty}_{k=0} \subseteq \cU$ is a sequence of approximate solutions 
generated by some method for solving \eqnok{dual}. It follows from Proposition \ref{pd-soln} 
that given any $\epso > 0$, there exists some iterate $U_k$ such that $g_{\rho}(U_k) 
- f_{\rho}(X(U_k)) \le \epso$. Then, it is clear that $X(U_k)$ is an $\epso$-optimal 
solution of \eqnok{primal} and hence \eqnok{penalty-prob1}. We next discuss two first 
order methods, namely, the adaptive spectral projected gradient method and the adaptive 
Nesterov's smooth method for problems \eqnok{dual} and \eqnok{primal} (or equivalently, 
\eqnok{penalty-prob1}).

\subsubsection{Adaptive spectral gradient projection method}
\label{aspg}

In this subsection, we propose an adaptive spectral projected gradient (ASPG) method for 
solving problems \eqnok{dual} and \eqnok{primal} (or equivalently, \eqnok{penalty-prob1}). 

The spectral gradient projection (SPG) methods were developed by Birgin et al.\ 
\cite{BiMaRa00} for minimizing a smooth function over a closed convex set, which well  
integrate the nonmonotone line search technique proposed by Grippo et al.\ \cite{GrLaLu86} 
and Barzilai-Borwein's gradient method \cite{BarBor83} into classical projected gradient
methods (see \cite{Bert99}). We next discuss the one of them (namely, the SPG2 method 
\cite{BiMaRa00}) for solving the problem
\beq \label{gbeta-prob}
\min \ \{g_{\rho,\beta}(U): \ U\in\cU \},
\eeq
and its dual 
\beq \label{frho-prob}
\max \ \{f_{\rho}(X): \ \alpha_{\rho} I \preceq X \preceq \beta I\}
\eeq
for some $\beta \ge \alpha_{\rho}$, where 
\beq \label{grho-beta}
g_{\rho,\beta}(U) := \max\limits_{\arho I \preceq X  \preceq \beta I} \ \phi(X,U),
\eeq
$\cU$, $\phi(\cdot,\cdot)$, $f_{\rho}(\cdot)$ and $\arho$ are defined in \eqnok{cU}, 
\eqnok{phi}, \eqnok{primal} and Proposition \ref{soln-bdd}, respectively. We denote by 
$X_{\beta}(U)$ the unique optimal solution of problem \eqnok{grho-beta}. In view 
of \eqnok{phi}, it is not hard to observe that $g_{\rho,\beta}(U)$ is differentiable, and, 
moreover, $X_{\beta}(U)$ and $\nabla g_{\rho,\beta}(U)$ have closed-form expressions for 
any $U\in\cU$ (see (30) of \cite{Lu07}). In addition, since $\cU$ is a simple set, the 
projection of a point to $\cU$ can be cheaply carried out. Thus, the SPG method  
\cite{BiMaRa00} is suitable for solving problem \eqnok{gbeta-prob}. 

For ease of subsequent presentation, we now describe the SPG method  
\cite{BiMaRa00} for \eqnok{gbeta-prob} in details. The following notation will be used 
throughout this subsection.  

Given a sequence $\{U_k\}^{\infty}_{k=0} \subseteq \cU$ and an integer $M \ge 1$, we define  
\[
g^M_k := \max \ \{g_{\rho,\beta}(U_{k-j}): \ 0 \le j \le \min\{k, M-1\}\}.
\] 
Also, let $P_{\cU}: \Re^{n \times n} \to \cU$ be defined as 
\[
P_{\cU}(U) := \arg\min\{\|\hat U - U\|_F: \ \hat U \in \cU \}, \ \forall U \in \Re^{n \times n}. 
\]
 
\vgap
 
\noindent
\begin{minipage}[h]{6.6 in}
{\bf The SPG method for problems \eqnok{gbeta-prob} and \eqnok{frho-prob}:} \\ [5pt]
Let $\epso>0$, $\gamma \in (0,1)$, $0 < \sigma_1 < \sigma_2 < 1$ and $0 < \amin < 
\amax < \infty$ be given. Let $M \ge 1$ be an integer. Choose $U_0 \in \cU$, $\alpha_0 \in 
[\amin, \amax]$ and set $k=0$.

\begin{itemize}
\item[1)] If $g_{\rho,\hbeta}(U_k) - f_{\rho}(\xhbuk) \le \epso$, terminate.
\item[2)] Compute $d_k = P_{\cU}(U_k - \alpha_k \nabla g_{\rho,\beta}(U_k)) - U_k$. Set 
$\lambda \leftarrow 1$.
\bi
\item[2a)] Set $U_+ = U_k + \lambda d_k$.   
\item[2b)] If $g_{\rho,\beta}(U_+) \le g^M_k + \gamma \lambda \langle d_k, 
\nabla g_{\rho,\hbeta}(U_k) \rangle$, set $U_{k+1} = U_+$, $s_k = U_{k+1}-U_k$, 
$y_k = \nabla g_{\rho,\hbeta}(U_{k+1}) - \nabla g_{\rho,\hbeta}(U_k)$. Otherwise, choose 
$\lambda_+ \in [\sigma_1\lambda, \sigma_2\lambda]$, set $\lambda \leftarrow \lambda_+$ and 
go to step 2a).  
\item[2c)] Compute $b_k = \langle s_k, y_k \rangle$. If $b_k \le 0$, set $\alpha_{k+1} = 
\amax$. Otherwise, compute $a_k = \langle s_k, s_k \rangle$ and set 
$\alpha_{k+1} = \min\ \{\amax, \max\{\amin, a_k/b_k\}\}$.
\ei
\item[3)]
Set $k \leftarrow k+1$, and go to step 1). 
\end{itemize}
\noindent
{\bf end}
\end{minipage}

\gap
 
We next establish a convergence result for the SPG method for solving problems 
\eqnok{gbeta-prob} and \eqnok{frho-prob}.

\begin{theorem} \label{spg2-converg}
Let $\epso > 0$  be given. The SPG method generates a pair of $\epso$-optimal solutions 
$(U_k, X_{\beta}(U_k))$ to problems \eqnok{gbeta-prob} and \eqnok{frho-prob} in a finite 
number of iterations.
\end{theorem}

\begin{proof}
Suppose by contradiction that the SPG method does not terminate. Then it generates a sequence 
$\{U_k\}^{\infty}_{k=0} \subseteq \cU$ satisfying $g_{\rho,\hbeta}(U_k) - f_{\rho}(\xhbuk) > 
\epso$. Note that $g_{\rho,\beta}(\cdot)$ is convex, which together with Theorem 2.4 of 
\cite{BiMaRa00} implies that any accumulation point of $\{U_k\}^{\infty}_{k=0}$ is an 
optimal solution of problem \eqnok{gbeta-prob}. By the continuity of $g_{\rho,\beta}(\cdot)$, 
it further implies that any accumulation point of $\{g_{\rho,\beta}(U_k)\}^{\infty}_{k=0}$ is 
the optimal value $f^*_{\rho}$ of \eqnok{gbeta-prob}. Using this observation and the fact that 
$\{g_{\rho,\beta}(U_k)\}^{\infty}_{k=0}$ is bounded, we conclude that $g_{\rho,\beta}(U_k) \to 
f^*_{\rho}$ as $k \to \infty$. Further, in view of Proposition \ref{pd-soln} by replacing 
$\brho$ with $\beta$, and $g_{\rho}(\cdot)$ with $g_{\rho,\beta}(\cdot)$, we have 
$g_{\rho,\hbeta}(U_k) - f_{\rho}(\xhbuk) \to 0$ as $k \to \infty$, and arrive at a contradiction. 
Therefore, the conclusion of this theorem holds.
\end{proof}

\vgap

Based on the above discussion, we see that the SPG method can be directly applied to find a 
pair of $\epso$-optimal solutions to problems \eqnok{dual} and \eqnok{primal} (or equivalently, 
\eqnok{penalty-prob1}) by setting $\beta = \beta_{\rho}$, where $\beta_{\rho}$ is given in 
Proposition \ref{soln-bdd}. It may converge, however, very slowly when $\beta_{\rho}$ is 
large. Indeed, similarly as in \cite{Lu07}, one can show that $\nabla g_{\rho,\beta}(U)$ is 
Lipschitz continuous on $\cU$ with constant $L = \beta^2 (\max\limits_{i,j} \rho_{ij})^2$ 
with respect to the Frobenius norm. Let $\alpha_k$, $b_k$ and $d_k$ be defined as above. 
Since $g_{\rho,\beta}(\cdot)$ is convex, we have $b_k \ge 0$. Actually, we observed that 
it is almost always positive. In addition, $\amin$ and $\amax$ are usually set to be $10^{-30}$ 
and $10^{30}$, respectively. Thus for the SPG method, we typically have  
\[
\alpha_{k+1} = \frac{\|U_{k+1}-U_k\|^2_F}{\langle U_{k+1}-U_k, \nabla g_{\rho,\hbeta}(U_{k+1})
-\nabla g_{\rho,\hbeta}(U_k) \rangle} \ \ge \ = \ \frac1L \ = \ \frac{1}{\beta^2 (\max\limits_{i,j} 
\rho_{ij})^2}.
\]  
Recall that $\beta_{\rho}$ is an upper bound of $\lmax(X^*_{\rho})$, and typically it is overly 
large, where $X^*_{\rho}$ is the optimal solution of \eqnok{penalty-prob1}. When 
$\beta=\beta_{\rho}$, we see from above that $\alpha_{k}$ can be very small and so is 
$U_{k+1}-U_k$ due to 
\[
\|U_{k+1}-U_k\|_F \ \le \ \|d_k\|_F \ = \ \|P_{\cU}(U_k - \alpha_k \nabla g_{\rho,\beta}(U_k)) - U_k\|_F 
\ \le \ \alpha_k \|\nabla g_{\rho,\beta}(U_k)\|_F. 
\]
Therefore, the SPG method may converge very slowly when applied to problem \eqnok{dual} directly.

To alleviate the aforementioned computational difficulty, we next propose an adaptive SPG (ASPG) 
method for problems \eqnok{dual} and \eqnok{primal} (or equivalently, \eqnok{penalty-prob1}) by 
solving a sequence of problems \eqnok{gbeta-prob} with $\beta = \beta_0$, $\beta_1$, $\ldots$, 
$\beta_m$ for some $\{\beta_k\}^m_{k=0}$ approaching $\lmax(X^*_{\rho})$ monotonically from below. 

\gap
 
\noindent
\begin{minipage}[h]{6.6 in}
{\bf  The adaptive SPG (ASPG) method for problems \eqnok{primal} and \eqnok{dual}:} \\ [5pt]
Let $\epso>0$, $\beta_0 \ll \beta_{\rho}$ and $r_{\beta} > 1$ be given. 
Choose $U_0 \in \cU$ and set $k=0$.
\begin{itemize}
\item[1)] Set $\beta \leftarrow \beta_k$. Apply the SPG method to find a pair of 
$\epso$-optimal solutions $(\hat U_k, X_{\beta}(\hat U_k))$ to problems \eqnok{gbeta-prob} 
and \eqnok{frho-prob} starting from $U_0$.
\item[2)] If $\beta=\beta_{\rho}$ or $\lmax(X_{\beta}(\hat U_k)) < \beta$, terminate. 
\item[3)] Set $U_0 \leftarrow \hat U_k$, $\beta_{k+1} = \min\{\beta r_{\beta}, \beta_{\rho}\}$, 
$k \leftarrow k+1$, and go to step 1).
\end{itemize}
\noindent
{\bf end}
\end{minipage}

\vgap

We now establish a convergence result for the ASPG method for solving problems 
\eqnok{dual} and \eqnok{primal} (or equivalently, \eqnok{penalty-prob1}).

\begin{theorem} \label{adapt-spg2}
Let $\epso >0$ be given. The ASPG method generates a pair of $\epso$-optimal 
solutions to problems \eqnok{dual} and \eqnok{primal} (or equivalently, \eqnok{penalty-prob1}) 
in a finite number of total (inner) iterations. 
\end{theorem}

\begin{proof}
First, we clearly see that $\beta$ is updated for only a finite number of times. Using 
this observation and Theorem \ref{spg2-converg}, we conclude that the ASPG method 
terminates in a finite number of total (inner) iterations. Now, suppose that it terminates 
at $\beta=\beta_k$ for some $k$. We claim that $(\hat U_k, X_{\beta}(\hat U_k))$ is 
a pair of $\epso$-optimal solutions to problems \eqnok{dual} and \eqnok{primal} 
(or equivalently, \eqnok{penalty-prob1}). Indeed, we clearly have $\beta=\beta_{\rho}$ or 
$\lmax(X_{\beta}(\hat U_k)) < \beta$, which together with the definition of $g_{\rho}(\cdot)$ 
and $g_{\rho,\beta}(\cdot)$ (see \eqnok{dual} and \eqnok{gbeta-prob}), implies that 
$g_{\rho}(\hat U_k) = g_{\rho,\beta}(\hat U_k)$. Thus, we obtain that  
\[
g_{\rho}(\hat U_k) - f_{\rho}(X_{\beta}(\hat U_k)) \ = \ g_{\rho,\beta}(\hat U_k) - 
f_{\rho}(X_{\beta}(\hat U_k)) \ \le \ \epso,
\]
which along with the fact $X_{\beta}(\hat U_k) \in \cX_{\rho}$, implies that 
$(\hat U_k, X_{\beta}(\hat U_k))$ is a pair of $\epso$-optimal solutions to problems 
\eqnok{dual} and \eqnok{primal}.
\end{proof}

\gap

As discussed above, the ASPG method is able to find a pair of $\epso$-optimal solutions 
to problems \eqnok{penalty-prob1} and \eqnok{dual}. We now show how this method can be 
extended to find an $(\epso,\eps_c)$-optimal solution to problem \eqnok{gsparcov}. Recall 
from the framework GSICS (see Subsection \ref{framework}) that in order to obtain an 
$(\epso,\eps_c)$-optimal solution to problem \eqnok{gsparcov}, we need to find an 
$\epso$-optimal solution of problem \eqnok{penalty-prob1} for a sequence of penalty 
parameters $\{\rho^k\}^m_{k=1}$, which satisfy for $k=1,\ldots,m$, $\rho^k_{ij}=\rho_{ij}$, 
$\forall (i,j)\not\in\Omega$ and $\rho^k_{ij} = \rho^0_{ij} r^{k-1}_{\rho}$, $\forall (i,j)
\in\Omega$ for some $r_{\rho}>1$ and $\rho^0_{ij}>0$, $\forall (i,j)\in\Omega$. Suppose 
that a pair of $\epso$-optimal solutions $(X_{\beta_k}(\hat U_k)), \hat U_k)$ of problems 
\eqnok{penalty-prob1} and \eqnok{dual} with $\rho=\rho^k$ are already found by the ASPG 
method for some $\beta_k \in [\alpha_{\rho^k}, \beta_{\rho^k}]$. Then, we choose the 
initial $U_0$ and $\beta_0$ for the ASPG method when applied to solve problems 
\eqnok{penalty-prob1} and \eqnok{dual} with $\rho=\rho^{k+1}$ as follows: 
\beq \label{initial-pt}
(U_0)_{ij} = \left\{\ba{ll} 
(\hat U_k)_{ij}/r_{\rho}, & \ \mbox{if} \ (i,j) \in \Omega; \\
(\hat U_k)_{ij}, & \ \mbox{otherwise}.
\ea\right., \ \ \ \ \ \ \beta_0 = \max\left\{\alpha_{\rho^{k+1}},\lmax(X_{\beta_k}
(\hat U_k))\right\}.
\eeq  
We next provide some interpretation on such a choice of $U_0$ and $\beta_0$. Since 
$\hat U^k \in \cU$ and $r_{\rho} >1$, we easily see that $U_0 \in \cU$. In addition, 
using the definition of $\beta_{\rho}$ (see Proposition \ref{soln-bdd}) and the fact 
that $\Diag(\rho^{k+1})= \Diag(\rho^k)$, we observe that $\beta_{\rho^{k+1}}=\beta_{\rho^k}$, 
and hence $\beta_0 \in [\alpha_{\rho^{k+1}},\beta_{\rho^{k+1}}]$. Let $f^*_{\rho}$ denote 
the optimal value of problem \eqnok{penalty-prob1} for any given $\rho$. Clearly, we can 
observe from the ASPG method that either $\lmax(X_{\beta_k}(\hat U_k)) < \beta_k < 
\beta_{\rho^k}$ or $\lmax(X_{\beta_k}(\hat U_k) \le  \beta_k = \beta_{\rho^k}$ holds, which 
together with \eqnok{dual} and \eqnok{grho-beta} implies that  
\beq \label{optval-ineq}
g_{\rho^k, \beta_k}(\hat U_k)  = g_{\rho^k}(\hat U_k) \in [f^*_{\rho^k}, f^*_{\rho^k} + \epso].
\eeq
Typically, $\lmax(X_{\beta_k}(\hat U_k) \gg \alpha_{\rho^{k+1}}$, and hence 
$\beta_0 = \lmax(X_{\beta_k}(\hat U_k)) \le \beta_k$ generally holds. Also usually, 
$\alpha_{\rho^{k+1}} \approx \alpha_{\rho^k} \approx 0$. Using these relations along with 
\eqnok{optval-ineq}, \eqnok{dual} and \eqnok{grho-beta}, we further observe that 
\[
\ba{lcllcllcllcl}
f^*_{\rho^{k+1}} &\le & g_{\rho^{k+1}}(U_0) & \approx & g_{\rho^k} (\hat U_k) &\le & 
f^*_{\rho^k} + \epso, & &  \\ [6pt]
f^*_{\rho^{k+1}} &\le & g_{\rho^{k+1},\beta_0}(U_0) & \approx & g_{\rho^{k},\beta_0}(U_0) 
& \le & g_{\rho^k,\beta_k} (\hat U_k) & \le & f^*_{\rho^k} + \epso.
\ea
\]  
It follows that when $f^*_{\rho^{k+1}}$ is close to $f^*_{\rho^k}$, $U_0$ is nearly an
$\epso$-optimal solution for problems \eqnok{dual} and $\eqnok{grho-beta}$ with 
$\rho=\rho^{k+1}$ and $\beta=\beta_0$. Therefore, we expect that for the above choice 
of $U_0$ and $\beta_0$, the ASPG method can solve problems \eqnok{penalty-prob1} and 
\eqnok{dual} with $\rho=\rho^{k+1}$ rapidly when $\rho^{k+1}$ is close to $\rho^{k}$.

\subsubsection{Adaptive Nesterov's smooth method}
\label{ans}

In this subsection, we propose an adaptive Nesterov's smooth (ANS) method for solving 
problems \eqnok{dual} and \eqnok{primal} (or equivalently, \eqnok{penalty-prob1}).

Recently, Lu \cite{Lu07} studied Nesterov's smooth method \cite{Nest83-1,Nest05-1} 
for solving a special class of problems \eqnok{dual} and \eqnok{primal} 
(or equivalently, \eqnok{penalty-prob1}), where $\rho$ is a positive multiple of $ee^T$. 
He showed that an $\epso$-optimal solution to problems \eqnok{dual} and \eqnok{primal} 
can be found in at most $\sqrt{2} \beta_{\rho} (\max\limits_{i,j} \rho_{ij}) 
\max\limits_{U\in\cU}\|U-U_0\|_F/\sqrt{\epso}$ iterations by Nesterov's smooth method 
for some initial point $U_0\in\cU$ (see pp.\ 12 of \cite{Lu07} for details). Given that 
$\brho$ is an estimate and typically an overestimate of $\lmax(X^*_{\rho})$, 
where $X^*_{\rho}$ is the unique optimal solution of problem \eqnok{penalty-prob1}, the 
aforementioned iteration complexity can be exceedingly large and Nesterov's smooth 
method generally converges extremely slowly. Lu \cite{Lu07} further proposed an 
adaptive Nesterov's smooth (ANS) method for solving problems \eqnok{dual} and 
\eqnok{primal} (see pp.\ 15 of \cite{Lu07}). In his method, $\lmax(X^*_{\rho})$ is 
estimated by $\lmax(X(U_k))$ and adaptively adjusted based on the change of 
$\lmax(X(U_k))$ as the algorithm progresses, where $U_k$ is an approximate solution 
of problem \eqnok{dual}. As a result, his method can provide an asymptotically tight 
estimate of $\lmax(X^*_{\rho})$ and it has an asymptotically optimal iteration complexity.

We now extend the ANS method \cite{Lu07} to problems \eqnok{dual} and \eqnok{primal} 
(or equivalently, \eqnok{penalty-prob1}) with a general $\rho$. Recall from Subsection 
\ref{aspg} that $\nabla g_{\rho}(U)$ is Lipschitz continuous on $\cU$ with constant 
$L = \beta^2_{\rho} (\max\limits_{i,j} \rho_{ij})^2$ with respect to the Frobenius 
norm. Then it is straightforward to extend the ANS method \cite{Lu07} to problems 
\eqnok{dual} and \eqnok{penalty-prob1} for a general $\rho$ by replacing the corresponding 
Lipschitz constants by the ones computed according to the above formula. For ease of 
reference, we provide the details of the ANS method for problems \eqnok{dual} and 
\eqnok{primal} (or equivalently, \eqnok{penalty-prob1}) below.

Throughout the remainder of this section, we assume that $\arho$, $\brho$, $g_{\rho,\beta}
(\cdot)$ and $X_{\beta}(\cdot)$ are given in Proposition \ref{soln-bdd} and Subsection 
\ref{aspg}, respectively. We now introduce a definition that will be used subsequently.
 
\begin{definition} 
Given any $U\in \cU$ and $\hbeta \in [\arho,\brho]$, $X_{\hbeta}(U)$ is called ``active'' 
if $\lambda_{\max}(X_{\hbeta}(U))=\hbeta$ and $\hbeta < \brho$; otherwise it is 
called ``inactive''.
\end{definition}

We are now ready to present the ANS method \cite{Lu07} for problems \eqnok{dual} and 
\eqnok{primal}.

\gap

\noindent
\begin{minipage}[h]{6.6 in}
{\bf The ANS method for problems \eqnok{primal} and \eqnok{dual}} \\ [5pt]
Let $\epsilon > 0$, $\varsigma_1$, $\varsigma_2 > 1$, and let $\varsigma_3 \in (0,1)$ be 
given. Let $\rho_{\max} = \max\limits_{i,j} \rho_{ij}$. Choose $U_0 \in \cU$ and 
$\hbeta \in [\arho, \brho]$. Set $L=\hbeta^2 \rho^2_{\max}$, $\sigma=1$, and $k=0$.
\begin{itemize}
\item[1)] Compute $\xhbuk$.
\bi
\item[1a)] 
If $\xhbuk$ is active, find the smallest $s\in \cZ_{+}$ such that $X_{\bbeta}(U_k)$ is 
inactive, where $\bbeta=\min\{\varsigma_1^s \hbeta, \brho\}$. Set $k=0$, $U_0 = U_k$, 
$\hbeta = \bbeta$, $L=\hbeta^2 \rho^2_{\max}$ and go to step 2). 
\item[1b)]
If $\xhbuk$ is inactive and $\lmax(\xhbuk) \le \varsigma_3 \hbeta$, set $k=0$, $U_0 = U_k$, 
\\ $\hbeta=\max\{\min\{\varsigma_2 \lmax(\xhbuk),\brho\},\arho\}$, and 
$L=\hbeta^2\rho^2_{\max}$. 
\ei
\item[2)] If $g_{\rho,\hbeta}(U_k) - f_{\rho}(\xhbuk) \le \epsilon$, terminate.
Otherwise, compute $\nabla g_{\rho,\hbeta}(U_k)$.
\item[3)]
Find $U^{sd}_{k} = \argmin \left \{ \langle \nabla g_{\rho,\hbeta}(U_k), U-U_k \rangle +
\frac{L}2 \, \|U-U_k\|^2_F: \ U \in \cU \right \}$.
\item[4)]
Find $U^{ag}_{k} = \argmin \left \{ \frac{L}{2\sigma}\|U-U_0\|^2_F+\sum\limits_{i=0}^k
\frac{i+1}2 [g_{\rho,\hbeta}(U_i)+ \langle \nabla g_{\rho,\hbeta}(U_i), U-U_i \rangle]: 
\ U \in \cU \right \}$. 
\item[5)]
Set $U_{k+1} = \frac{2}{k+3} U^{ag}_{k} + \frac{k+1}{k+3} U^{sd}_{k}$.
\item[6)]
Set $k \leftarrow k+1$, and go to step 1). 
\end{itemize}
\noindent
{\bf end}
\end{minipage}

\vgap


Similarly as the ASPG method, we can easily extend the ANS method to find an 
$(\epso,\eps_c)$-optimal solution to problem \eqnok{gsparcov} by applying the same strategy 
for updating the initial $U_0$ and $\beta_0$ detailed at the end of Subsection \ref{aspg}. 
For convenience of presentation, the resulting method is referred to as the adaptive Nesterov's 
smooth (ANS) method. 

 
\section{Computational results}
\label{comp}

In this section, we test the sparse recovery ability of the model \eqnok{gsparcov} and 
compare the performance of the adaptive spectral projected gradient (ASPG) method and the 
adaptive Nesterov's smooth (ANS) method that are proposed in Section \ref{methods} for 
solving problem \eqnok{gsparcov} on a set of randomly generated instances. 

All instances used in this section were randomly generated in a similar manner as 
described in d'Aspremont et al.\ \cite{DaOnEl06} and Lu \cite{Lu07}. Indeed, we first 
generate a sparse matrix $A\in \cS^n_{++}$, and then we generate a matrix $B\in\cS^n$ by
\[
B = A^{-1} + \tau V,
\]      
where $V\in\cS^n$ contains pseudo-random values drawn from a uniform distribution on 
the interval $[-1,1]$, and $\tau$ is a small positive number. Finally, we obtain 
the following randomly generated sample covariance matrix:
\[
\Sigma = B - \min\{\lambda_{\min}(B)- \vartheta, 0\} I,   
\]
where $\vartheta$ is a small positive number. In particular, we set $\tau=0.15$, 
$\vartheta=1.0e-4$ for generating all instances.    
 

In the first experiment we compare the performance of the ASPG and ANS methods for 
problem \eqnok{gsparcov}. For this purpose, we first randomly generate the above 
matrix $A\in \cS^n_{++}$ with a density prescribed by $\varrho$, and set $\Omega = \{(i,j): 
\ A_{ij}=0, |i-j| \ge 2\}$ and $\rho_{ij}=0.5$ for all $(i,j)\notin\Omega$. $\Sigma$ is then 
generated by the above approach. The codes for both methods are written in 
MATLAB. In particular, we set $\gamma = 10^{-4}$, $M=8$, $\sigma_1=0.1$, $\sigma_2=0.9$, 
$\alpha_{\min}=10^{-15}$, $\alpha_{\max}=10^{15}$ for the ASPG method, and set 
$\varsigma_1= \varsigma_2 = 1.05$ and $\varsigma_3=0.95$ for the ANS method. 
In addition, for both methods we set $\beta_0=1$, $r_{\beta}=10$, $r_{\rho}=2$, and 
$\rho^0_{ij} = 0.5$ for all $(i,j) \in \Omega$. Also, the ASPG and ANS methods start 
from the initial point $U_0=0$ and terminate once an $(\epso,\eps_c)$-optimal solution 
of problem \eqnok{gsparcov} is found, where $\epso=0.1$ and $\eps_c=10^{-4}$. All 
computations are performed on an Intel Xeon 2.66 GHz machine with Red Hat Linux version 8.

The performance of the ASPG and ANS methods for the randomly generated instances with 
density $\varrho=0.1$, $0.5$ and $0.9$ is presented in Tables \ref{result-1}-\ref{result-3}, 
respectively. The row size $n$ of each sample covariance matrix  $\Sigma$ is given in column 
one. The size of the set $\Omega$ is given in column two. The numbers of (inner) iterations 
of ASPG and ANS are given in columns three to four, the number of function evaluations 
are given in columns five to six, and the CPU times (in seconds) are given in the last two 
columns, respectively. 
\begin{table}[t]
\caption{Comparison of ASPG and ANS for $\varrho=0.1$}
\centering
\label{result-1}
\begin{tabular}{|rr||rr||rr||rr|}
\hline 
\multicolumn{2}{|c||}{Problem} & \multicolumn{2}{c||}{Iter} &  
\multicolumn{2}{c||}{Nf} & \multicolumn{2}{c|}{Time} \\
\multicolumn{1}{|c}{n} & \multicolumn{1}{c||}{size($\Omega$)} & 
\multicolumn{1}{c}{\sc ans} & \multicolumn{1}{c||}{\sc aspg} & 
\multicolumn{1}{c}{\sc ans} & \multicolumn{1}{c||}{\sc aspg} & 
\multicolumn{1}{c}{\sc ans} & \multicolumn{1}{c|}{\sc aspg} \\
\hline
100 & 8792 & 1298 & 1736 & 1298 & 2626  & 17.9 & 33.9 \\
200 & 35646 & 593 & 489 & 593 & 654 & 52.1 & 56.1 \\
300 & 80604 & 1411 & 683 & 1411 & 974 & 431.8 & 291.7 \\
400 & 143636 & 1400 & 702 & 1400 & 978 & 1053.8 & 730.4 \\
500 & 224788 & 1012 & 615 & 1012 & 863 & 1469.4 & 1244.8 \\
600 & 324072 & 1410 & 661 & 1410 & 908 & 3501.2 & 2220.5 \\
700 & 441380 & 1189 & 738 & 1189 & 1050 & 4656.0 & 4070.5 \\
800 & 576896 & 1175 & 811 & 1175 & 1169 & 6601.2 & 6500.5 \\
900 & 730500 & 1660 & 808 & 1660 & 1154 & 12975.7 & 8964.5 \\
1000 & 902124 & 2600 & 1285 & 2600 & 1903 & 27523.2 & 20059.9 \\
\hline
\end{tabular}
\end{table}
\begin{table}[t]
\caption{Comparison of ASPG and ANS for $\varrho=0.5$}
\centering
\label{result-2}
\begin{tabular}{|rr||rr||rr||rr|}
\hline 
\multicolumn{2}{|c||}{Problem} & \multicolumn{2}{c||}{Iter} &  
\multicolumn{2}{c||}{Nf} & \multicolumn{2}{c|}{Time} \\
\multicolumn{1}{|c}{n} & \multicolumn{1}{c||}{size($\Omega$)} & 
\multicolumn{1}{c}{\sc ans} & \multicolumn{1}{c||}{\sc aspg} & 
\multicolumn{1}{c}{\sc ans} & \multicolumn{1}{c||}{\sc aspg} & 
\multicolumn{1}{c}{\sc ans} & \multicolumn{1}{c|}{\sc aspg} \\
\hline
100 & 4776 & 256 & 112 & 256 & 146  & 3.9 & 2.3 \\
200 & 19438 & 453 & 178 & 453 & 229 & 40.3 & 20.1 \\
300 & 44136 & 412 & 229 & 412 & 296 & 128.4 & 91.4 \\
400 & 78738 & 433 & 250 & 433 & 339 & 335.1 & 260.2 \\
500 & 123300 & 499 & 313 & 499 & 417 & 727.2 & 605.8 \\
600 & 177614 & 535 & 354 & 535 & 494 & 1361.0 & 1247.1 \\
700 & 241944 & 569 & 327 & 569 & 467 & 2204.8 & 1793.9 \\
800 & 317184 & 536 & 349 & 536 & 498 & 3011.7 & 2763.5 \\
900 & 400952 & 581 & 420 & 581 & 600 & 4619.5 & 4752.2 \\
1000 & 494610 & 697 & 561 & 697 & 775 & 7425.6 & 8240.1 \\
\hline
\end{tabular}
\end{table}
\begin{table}[t]
\caption{Comparison of ASPG and ANS for $\varrho=0.9$}
\centering
\label{result-3}
\begin{tabular}{|rr||rr||rr||rr|}
\hline 
\multicolumn{2}{|c||}{Problem} & \multicolumn{2}{c||}{Iter} &  
\multicolumn{2}{c||}{Nf} & \multicolumn{2}{c|}{Time} \\
\multicolumn{1}{|c}{n} & \multicolumn{1}{c||}{size($\Omega$)} & 
\multicolumn{1}{c}{\sc ans} & \multicolumn{1}{c||}{\sc aspg} & 
\multicolumn{1}{c}{\sc ans} & \multicolumn{1}{c||}{\sc aspg} & 
\multicolumn{1}{c}{\sc ans} & \multicolumn{1}{c|}{\sc aspg} \\
\hline
100 & 960 & 207 & 85 & 207 & 164  & 3.3 & 2.5 \\
200 & 3738 & 275 & 139 & 275 & 180 & 24.5 & 16.0 \\
300 & 8750 & 567 & 178 & 567 & 220 & 173.4 & 69.7 \\
400 & 15764 & 408 & 180 & 408 & 235 & 318.6 & 182.7 \\
500 & 25072 & 416 & 272 & 416 & 367 & 616.8 & 535.2 \\
600 & 35846 & 441 & 275 & 441 & 371 & 1107.0 & 920.7 \\
700 & 48718 & 1219 & 421 & 1219 & 597 & 4646.2 & 2300.0 \\
800 & 63814 & 461 & 348 & 461 & 460 & 2693.9 & 2650.0 \\
900 & 80798 & 469 & 363 & 469 & 507 & 4124.1 & 4171.8 \\
1000 & 98870 & 495 & 363 & 495 & 514 & 5656.1 & 5718.9 \\
\hline
\end{tabular}
\end{table}
From Tables \ref{result-1}-\ref{result-3}, we see that both methods are able to solve 
all instances within a reasonable amount of time. In addition, the ASPG method, namely, 
the adaptive spectral gradient method, generally outperforms the ANS method, that is, 
the adaptive Nesterov's smooth method.    

Our second experiment is similar to the one carried out in d'Aspremont et al.\ \cite{DaOnEl06}.  
We intend to test the sparse recovery ability of the model \eqnok{gsparcov}. To this aim, 
we specialize $n=30$ and the matrix $A\in \cS^n_{++}$ to be the one with diagonal 
entries around one and a few randomly chosen, nonzero off-diagonal entries equal to 
$+1$ or $-1$ and the sample covariance matrix $\Sigma$ is then generated by the aforementioned 
approach. Also, we set $\Omega = \{(i,j): \ A_{ij}=0, |i-j| \ge 5\}$ and $\rho_{ij}=0.1$ for all 
$(i,j)\notin\Omega$. The model \eqnok{gsparcov} with such an instance is finally solved by the 
ASPG method whose parameters, initial point and termination criterion are exactly same as above. 
In Figure \ref{fig}, we plot the sparsity patterns of the original inverse covariance matrix $A$, 
the approximate solution to problem \eqnok{gsparcov} and the noisy inverse covariance matrix 
$B^{-1}$ for such a randomly generated instance. We observe that the model \eqnok{gsparcov} 
is capable of recovering the sparsity pattern of the original inverse covariance matrix.  
\begin{figure}
\vspace{-1.5in} 
\centering
\includegraphics[scale=1.0]{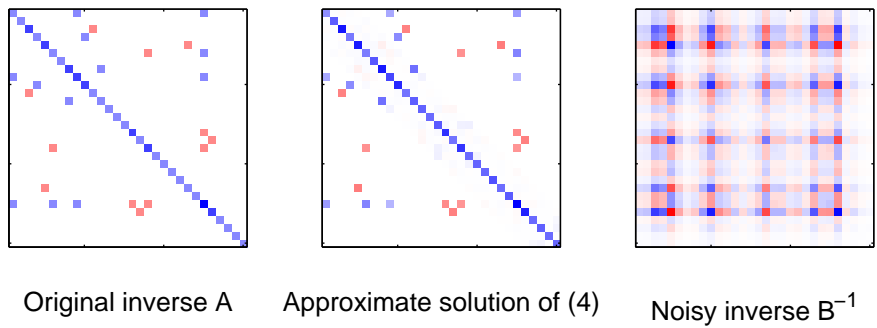} 
\vspace{-1.3in} 
\caption{Sparsity recovery.} \label{fig}
\end{figure}
    
\section{Concluding remarks}
\label{concl-remark}    

In this paper, we considered estimating sparse inverse covariance of a Gaussian 
graphical model whose conditional independence is assumed to be partially known. 
Naturally, we formulated it as a constrained $l_1$-norm penalized maximum likelihood 
estimation problem. Further, we proposed an algorithm framework, and developed two 
first-order methods, that is, adaptive spectral projected gradient (ASPG) 
method and adaptive Nesterov's smooth (ANS) method, for solving it. Our 
computational results demonstrate that both methods are able to solve problems of 
size at least a thousand and number of constraints of nearly a half million within 
a reasonable amount of time, and the ASPG method generally outperforms the ANS method. 


The source codes for the ASPG and ANS methods (written in MATLAB) are available 
online at www.math.sfu.ca/$\sim$zhaosong. They can also be applied to problem 
\eqnok{gsparcov} with $\Omega=\emptyset$, namely, the case where the underlying 
sparsity structure is completely unknown. It shall be mentioned that these codes 
can be extended straightforwardly to more general problems of the form
\[
\begin{array}{ll}
\max\limits_X & \log\det X - \langle \Sigma, X \rangle - \sum\limits_{(ij)\notin 
\Omega} \rho_{ij} |X_{ij}|  \\ 
\mbox{s.t.} & \alpha I \preceq X \preceq \beta I, \\ [4pt]
            & X_{ij} = 0, \ \forall (i,j) \in \Omega,  
\end{array}   
\]
where $0 \le \alpha < \beta \le \infty$ are some fixed bounds on the eigenvalues 
of the solution.


\begin{thebibliography}{10}

\bibitem{BarBor83} 
{\sc J.~Barzilai and J.~M.~Borwein}, {\em Two point step size gradient 
methods}, IMA J. Numer. Anal., 8 (1988), pp.~141--148.

\bibitem{Bert99}
{\sc D.~P.~Bertsekas}, {\em Nonlinear Programming}, 2nd edition, Athena 
Scientific, Belmont, Massachusetts, 1999. 

\bibitem{BiMaRa00}
{\sc E.~G.~Birgin, J.~M.~Mart\'{i}nez, and M.~Raydan}, {\em Nonmonotone 
spectral projected gradient methods on convex sets}, SIAM J. Optim., 
10 (2000), pp.~1196--1211.

\bibitem{DaVaRo06}
{\sc J.~Dahl, L.~Vandenberghe, and V.~Roychowdhury}, {\em Covariance 
selection for non-chordal graphs via chordal embedding}, Optim. Methods 
Softw., 23 (2008), pp.~501--520.

\bibitem{DaOnEl06}
{\sc A.~{d}'Aspremont, O.~Banerjee, and L.~{El Ghaoui}}, 
{\em First-order methods for sparse covariance selection}, SIAM J. Matrix 
Anal. Appl., 30 (2008), pp.~56--66.

\bibitem{FriHasTib07}
{\sc J.~Friedman, T.~Hastie, and R.~Tibshirani}, {\em Sparse inverse 
covariance estimation with the graphical lasso},  Biostatistics, 9 (2008), 
pp.~432--441.

\bibitem{GrLaLu86} 
{\sc L.~Grippo, F.~Lampariello, and S.~Lucidi}, 
{\em A nonmonotone line search technique for Newton's method}, 
SIAM J. Numer. Anal., 23 (1986), pp.~707--716.

\bibitem{Lu07}
{\sc Z.~Lu}, {\em Smooth optimization approach for sparse covariance selection}, 
SIAM J. Optim., to appear.

\bibitem{Nest83-1}
{\sc Y.~E. Nesterov}, {\em A method for unconstrained convex minimization 
problem with the rate of convergence {$O(1/k^2)$}}, Doklady AN SSSR, 269 
(1983), pp.~543--547, translated as Soviet Math. Docl.

\bibitem{Nest05-1}
{\sc Y.~E. Nesterov}, {\em Smooth minimization of nonsmooth functions}, 
Math. Programming, 103 (2005), pp.~127--152.

\bibitem{NeNe94}
{\sc Y.~E. Nesterov and A.~S. Nemirovski}, {\em Interior point Polynomial 
algorithms in Convex Programming: Theory and Applications}, SIAM, 
Philadelphia, 1994.

\bibitem{VaBoWu98}
{\sc L.~Vandenberghe, S.~Boyd, and S.~Wu}, {\em Determinant maximization 
with linear matrix inequality constraints}, SIAM J. Matrix Anal. Appl., 19 
(1998), pp.~499--533.

\bibitem{YuLi07-1}
{\sc M.~Yuan and Y.~Lin}, {\em Model selection and estimation in the 
{G}aussian graphical model}, Biometrika, 94 (2007), pp.~19--35.

\end{thebibliography}
\end{document}